\documentclass[11pt,a4paper]{article}

\usepackage[utf8]{inputenc}
\usepackage[english]{babel}

\usepackage{geometry}
\usepackage{graphicx}
\usepackage{amssymb,amsmath,amsthm,amsfonts}
\usepackage{delarray}
\usepackage{enumitem}
\usepackage{tikz}
\usepackage{hyperref}
\usepackage{authblk}

\graphicspath{{figures/}}

\newtheorem{theorem}{\textbf{Theorem}}

\newtheorem{remark*}{Remark}

\newtheorem{claim}{\textbf{Claim}}

\newcommand{\N}{\mathbb{N}}
\newcommand{\Z}{\mathbb{Z}}

\newcommand{\Poly}{{\mathsf{P}}}

\newcommand{\NP}{{\mathsf{NP}}}
\newcommand{\coNP}{\text{co-}{\mathsf{NP}}}
\newcommand{\PSPACE}{{\mathsf{PSPACE}}}
\let\emptyset\varnothing 

\newcommand{\decisionpb}[3]{\fbox{\parbox{0.9\textwidth}{{\bf #1}\\{\it input:} #2\\{\it question:} #3}}}

\definecolor{csoldiers}{HTML}{9933ff} 
\definecolor{csquare}{HTML}{b3b300} 
\definecolor{ccontour}{HTML}{ff0066} 
\definecolor{cguide}{HTML}{e67300} 
\definecolor{carms}{HTML}{00b386} 
\definecolor{canchors}{HTML}{0055ff} 
\definecolor{citems}{HTML}{002080} 
\definecolor{czippers}{HTML}{804000} 
\definecolor{chigh}{HTML}{ff4dff} 

\newcommand{\tm}{$^\text{\tiny TM}$}

\newcommand{\hdomino}[2]{\raisebox{-4pt}{\tikz[scale=.5]{\draw[black!25] (1,0) -- ++ (0,1); \draw (0,0) rectangle ++ (2,1); \node at (.5,.5) {\scriptsize $#1$}; \node at (1.5,.5) {\scriptsize $#2$};}}}
\newcommand{\vdomino}[2]{\raisebox{-4pt}{\tikz[scale=.5]{\draw[black!25] (0,1) -- ++ (1,0); \draw (0,0) rectangle ++ (1,2); \node at (.5,.5) {\scriptsize $#1$}; \node at (.5,1.5) {\scriptsize $#2$};}}}
\newlength{\textone}
\newlength{\texttwo}
\newcommand{\hhdomino}[2]{
  \settowidth{\textone}{\scriptsize $#1$}
  \settowidth{\texttwo}{\scriptsize $#2$}
  \raisebox{-4pt}{\tikz[scale=.5]{
    \draw[black!25] (2*\the\textone+8,0) -- ++ (0,1);
    \draw (0,0) rectangle ++ (2*\the\textone+2*\the\texttwo+16,1);
    \node at (\the\textone+4,.5) {\scriptsize $#1$};
    \node at (2*\the\textone+\the\texttwo+12,.5) {\scriptsize $#2$};
  }}
}

\newcommand{\crown}{\tikz[scale=.05]{
  \fill (0,0) -- (0,2) -- (.5,.5) -- (1,2) -- (1.5,.5) -- (2,2) -- (2,0) -- cycle;
  \draw (0,2) circle (5pt);
  \draw (1,2) circle (5pt);
  \draw (2,2) circle (5pt);
}}

\title{$\NP$-completeness of the game {\em Kingdomino\tm}}
\author[2]{Viet-Ha Nguyen}
\author[1]{K\'evin Perrot}
\author[3]{Mathieu Vallet}
\affil[1]{Universit\'e publique}
\affil[2]{Univ. C{\^o}te d'Azur, CNRS, Inria, I3S, UMR 7271, Sophia Antipolis, France}
\affil[3]{Aix Marseille Univ., Univ. Toulon, CNRS, LIS, UMR 7020, Marseille, France}
\date{}

\begin{document}
\renewcommand{\labelitemi}{$\circ$}
\renewcommand{\labelitemii}{$\circ$}
\setlist[itemize,enumerate]{nosep}
\maketitle

\begin{abstract}
  {\em Kingdomino\tm} is a board game designed by Bruno Cathala and edited by
  {\em Blue Orange} since 2016. The goal is to place $2 \times 1$ dominoes
  on a grid layout,
  and get a better score than other players.
  Each $1 \times 1$ domino cell has a color
  that must match at least one adjacent cell,
  and an integer number of crowns (possibly none) used to compute the score.
  We prove that even with full knowledge of the future of the game,
  it is $\NP$-complete to decide whether a given score is
  achievable at {\em Kingdomino\tm}.
\end{abstract}

\section{Introduction}
\label{s:intro}

{\em Kingdomino\tm} is a 2-4 players game where players, turn by turn, place $2
\times 1$ dominoes on a grid layout (each player has its own board, independent
of others). Each domino has a color on each of its two $1 \times 1$ cells, and
when a player is given a domino to place on its board, he or she must do so
with a color match along at least one of its edges. Also, if a domino {\em can}
be placed (with at least one possible color match) then it {\em must} be
placed, and if it cannot be placed then it is discarded.
Finally, each player starts from a $1 \times 1$ tower matching any
color. The winner of the game is the player that has the maximum score among
the competitors. The computation of score will be precised in
Section~\ref{s:model}, it is basically a weighted sum of the number of cells in
each monochromatic connected components on the player's board.

The purpose of this article is to prove that, even with full knowledge of the
future of the game ({\em i.e.} the sequence of dominoes he or she will have to
place), a player willing to know whether a given score is achievable
is faced with an $\NP$-complete decision problem.

Section~\ref{s:games} reviews some results around games complexity and domino
problems, Section~\ref{s:model} presents our theoretical modeling of {\em
Kingdomino\tm}, Section~\ref{s:count} illustrates the combinatorial explosion
of possibilities, and Section~\ref{s:NPc} proves the $\NP$-hardness result.

\section{Computational complexity of games with dominoes}
\label{s:games}

{\em Kingdomino\tm} has been studied in~\cite{kingdomino18}, where the authors compare different strategies to play the game, via numerical simulations.

Understanding the computational complexity of games has raised some interest
in the computer science community, and numerous games have been proven to be
complete for $\NP$ or $\coNP$.
Examples include
{\em Minesweeper}~\cite{mine1,mine2},
{\em SET}~\cite{gameset},
{\em Hanabi}~\cite{hanabi},
some {\em Nintendo} games~\cite{nintendo} and
{\em Candy Crush}~\cite{candy1}.

Domino tiling problems are a cornerstone for computer
science, from undecidable ones~\cite{b66,jr15,k08}
to simple puzzles~\cite{h76,ku68,s10}.
Tiling some board with dominoes under constraints
has already been seen to be $\NP$-complete, and constructions 
vary according to the model definition~\cite{b08,er13,mr01,py13}.
The model of \cite{ww04} is close to {\em Kingdomino\tm}, its
construction can be adapted to prove that starting from a board
with some dominoes already placed, deciding whether it can be completed
to achieve at least a given score is $\NP$-complete.
The challenging part
of the present work is to start from nothing else but
the tower of {\em Kingdomino\tm}.

\section{Model and problem statement}
\label{s:model}

In order to apply the theory of
computational complexity to {\em Kingdomino\tm}, we will consider a one player model
which concentrates on an essential aspect of the game: how to
maximize one's score, which is the source of domino choices. Also, in the game
{\em Kingdomino\tm} there is a fixed set of colors (6) and a fixed multiset of
dominoes (48, some dominoes have more than one occurrence), but we will
abstract these quantities to be any finite (multi)set.

A {\em domino} is a $2 \times 1$ rectangle, with one {\em color} among a finite
set on each of its two $1 \times 1$ {\em cells}. A domino also has a number of {\em
crowns} on each of its cells, used to compute the score. For convenience we
consider colors to be integers, and represent a domino as follows:
$\hdomino{1{\crown}}{2}$ for the domino with one cell of color $1$ with one
crown, and one cell of color $2$ with no crown.
A {\em tiling} is an overlap free placement of dominoes on the $\Z^2$ board,
with a special cell at position $(0,0)$ called the {\em tower}.
Given a sequence of $n$ dominoes
$\tau=(\,\hdomino{c_1}{c_2}\,,\dots,\hhdomino{c_{2n-1}}{c_{2n}})$ possibly with
crowns, a {\em K-tiling by $\tau$} is a tiling respecting the following
constraints defined inductively.
\begin{itemize}
  \item The tiling with only the tower at position $(0,0)$ is a K-tiling by
    $\emptyset$ (case $n=0$).
  \item Given a K-tiling by the $n-1$ first dominoes of $\tau$, the last domino
    of colors $c_{2n-1}$ and $c_{2n}$ can be placed on a pair of adjacent
    positions, if and only if at least one of its two cells is adjacent to a
    cell of the same color that has already been placed, or is adjacent to the
    tower. It is lost (not placed) if and only if it can be placed nowhere.
    This gives a K-tiling by the $n$ dominoes of $\tau$.
\end{itemize}
Hence dominoes must be placed in the order given by the sequence $\tau$.
Note that the definition of K-tiling does not take into consideration the
crowns. They are only used to compute the score, as we will explain now.

The {\em score} of a K-tiling is the sum, for each monochromatic connected
component (called {\em region}), of its number of cells times the number of
crowns it contains. Note that a color may give rise to more than one region,
and that a {\em region} scores no
point if it contains no crown.
We will say that some cell (resp. domino) must be {\em
connected} to some region, to mean that it (resp. one of its two cells) must belong to
this monochromatic connected component.

\begin{figure}
  \centerline{\includegraphics{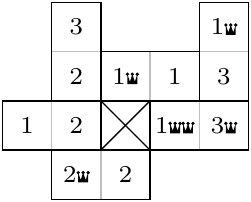}}
  \caption[Example of K-tiling, with score.]{Example of K-tiling by $(\,
  \hdomino{2{\crown}}{2}\,,
  \hdomino{1}{2}\,,
  \hdomino{1{\crown}}{1}\,,
  \hdomino{1{\crown\crown}}{3{\crown}}\,,
  \hdomino{3}{1{\crown}}\,,
  \hdomino{2}{3}\,
  )$, with score $0+1+4+9+2+0=16$
  (regions are ordered from top to bottom, left to right).}
  \label{fig:def}
\end{figure}

Definitions are
illustrated on Figure~\ref{fig:def}.
We are ready to state the problem.\\[.5em]
\decisionpb{K-tiling problem}
{a sequence of dominoes $\tau$ and an integer $s$.}
{is there a K-tiling by $\tau$ with score at least $s$?}\\[.5em]
Given a tiling where each domino of the sequence $\tau$ is identified (a
potential solution, {\em aka} certificate), one verifies domino after domino
that it is indeed a K-tiling by $\tau$, and computes the score to verify that
it is indeed at least $s$, in polynomial time. Hence {\bf K-tiling problem}
belongs to $\NP$.

Remark that our modeling of {\em Kingdomino\tm} discards the official game rule
regarding a bounding box for player's boards, where dominoes must be placed
inside a square of size $5 \times 5$ or $7 \times 7$ containing the tower.

\section{Counting K-tilings}
\label{s:count}

To give an idea of the combinatorial explosion one faces when playing
{\em Kingdomino\tm} or when deciding some {\bf K-tiling problem} instance, we
propose in Table~\ref{table:count} to count the number of possible K-tilings
for some small sequences of dominoes. These exact results were obtained by
brute force numerical simulations.

\begin{table}[!h]
  \centerline{\begin{tabular}{c|c|c|c|}
    dominoes &
    $(\,\vdomino{1{\crown}}{1},\vdomino{1}{1},\vdomino{1}{1},\vdomino{1}{1}\,)$ &
    $(\,\vdomino{1{\crown}}{1},\vdomino{2{\crown}}{2},\vdomino{3{\crown}}{3},\vdomino{4{\crown}}{4}\,)$ &
    $(\,\vdomino{1{\crown}}{2{\crown}},\vdomino{3{\crown}}{4{\crown}},\vdomino{5{\crown}}{6{\crown}},\vdomino{7{\crown}}{8{\crown}}\,)$\\[.2em]
    \hline
    $1^\text{st}$ (score 2) & 2 (24) & 2 (24) & 4 (24)\\
    $2^\text{nd}$ (score 4) & 19 (752) & 13 (400) & 52 (400)\\
    $3^\text{rd}$ (score 6) & 253 (35448) & 63 (4032) & 504 (4032)\\
    $4^\text{th}$ (score 8) & 3529 (2176064) & 141 (18048) & 2256 (18048)\\
  \end{tabular}}
  \caption{Number of K-tilings reaching the maximum possible score, for some
  small sequences of dominoes. Rotations and axial symmetries are counted only
  once, and the positions of crowns are not taken into account
  (in parenthesis the full counts are given).}
  \label{table:count}
\end{table}

Table~\ref{table:count} may be compared to the number of domino tilings of a
$2n \times 2n$ square, appearing in the {\em Online Encyclopedia of Integer
Sequences} under reference {\tt A004003} \cite{oeisA004003}: 1, 2, 36, 6728,
12988816, 258584046368, 53060477521960000, {\dots}

\section{$\NP$-hardness of {\bf K-tiling problem}}
\label{s:NPc}

In this section we prove the main result of the article.

\begin{theorem}
  \label{theorem:main}
  {\bf K-tiling problem} is $\NP$-hard.
\end{theorem}

\begin{proof}
  We make a polynomial time many-one reduction
  from {\bf 4-Partition problem}, which is
  known to be strongly $\NP$-complete \cite{gj79}.
  This is important, since we encode the instance of {\bf 4-Partition problem}
  in unary into an instance of {\bf K-tiling problem} (basically, with $28x$
  domino cells for each item of size $x$).\\[.5em]
  \decisionpb{4-Partition problem}
  {$n$ items of integer sizes $x_1,\dots,x_n$, and $m=\frac{n}{4}$ bins of size $k$,\\
  \phantom{\em input:} with $n$ a multiple of four, $x_i>0$ for all $i$, and $\sum_{i=1}^{n}x_i=km$.}
  {is it possible to pack{\footnotemark} the $n$ items into the $m$ bins,\\
  \phantom{\em question:} with exactly four items (whose sizes thus sum to $k$) per bin?}
  \footnotetext{Of course an item cannot be split.}\\[.5em]

  Given such an instance of {\bf 4-Partition problem}, we first multiply by
  $28$ all item and bin sizes (for technical reasons to be explained later) and
  consider the equivalent instance with $n$ items of strictly positive integer
  sizes $x_1 \leftarrow 28x_1,\dots,x_n \leftarrow 28x_n$ and $m$ bins of size
  $k \leftarrow 28k$ (for convenience we keep the initial notations with $x$
  and $k$). We then construct (in polynomial time from a unary encoding) the
  following sequence of dominoes $\tau$:\\[-.6em]
  \begin{enumerate}[itemsep=.5em]
    \item {\em guardians}\quad $\hdomino{1{\crown}}{1},\hdomino{2}{2},\hdomino{3}{3},\hdomino{4}{4},$
    \item {\em square}\quad $(\,\hdomino{1}{1}\,)^{18m^2-6},$
    \item {\em contour}\quad $\hdomino{1}{5{\crown}},\hdomino{5}{6{\crown}},\hdomino{6}{7{\crown}},\hdomino{7}{8{\crown}},\hdomino{8}{1},\hdomino{8}{1},\hdomino{8}{2},\hdomino{8}{9{\crown}},$\\[.5em]
      \phantom{\em contour}\quad $(\,\hdomino{9}{9}\,)^{9m-8},\hdomino{9}{10{\crown}},$\\[.5em]
      \phantom{\em contour}\quad $\hdomino{10}{11{\crown}},\hdomino{11}{12{\crown}},\hdomino{12}{13{\crown}},\hdomino{13}{14{\crown}},\dots,\hhdomino{3m+10}{3m+11{\crown}},$\\[.5em]
      \phantom{\em contour}\quad $\hhdomino{3m+11}{3m+12{\crown}},\hhdomino{3m+12}{5},$
    \item {\em guide}\quad $(\,\hdomino{6}{7}\,)^{18m^2+12m},$
    \item {\em arms}\quad $(\,\hdomino{10}{11}\,)^{\frac{k}{4}+2},(\,\hdomino{13}{14}\,)^{\frac{k}{4}+2},(\,\hdomino{16}{17}\,)^{\frac{k}{4}+2},\dots,(\,\hhdomino{3m+10}{3m+11}\,)^{\frac{k}{4}+2}$
    \item {\em anchors}\quad $(\,\hhdomino{12}{3m+13}\,)^2,(\,\hhdomino{15}{3m+13}\,)^2,\dots,(\,\hhdomino{3m+9}{3m+13}\,)^2,$
    \item {\em items}\quad for each $x_i$ we have $\hhdomino{3m+13}{3m+13+i{\crown}},(\,\hhdomino{3m+13+i}{3m+13+i}\,)^{\frac{x_i}{2}-1},$
    \item {\em zippers}\quad $\hhdomino{11}{3m+13+n+1{\crown}},\hhdomino{3m+13+n+1}{13},$\\[.5em]
      \phantom{\em zippers}\quad $\hhdomino{14}{3m+13+n+2{\crown}},\hhdomino{3m+13+n+2}{16},$\\[.5em]
      \phantom{\em zippers}\quad $\dots,$\\[.5em]
      \phantom{\em zippers}\quad $\hhdomino{3m+8}{4m+13+n{\crown}},\hhdomino{4m+13+n}{3m+10},$\\[-.6em]
  \end{enumerate}
  and the target score $s=72m^2+54m+\frac{k}{2}(3m+1)+1$. This is our instance
  of {\bf K-tiling problem}, with the idea of the reduction presented on
  Figure~\ref{fig:idea}. Let us now prove that there exists a packing of
  the $n$ items into the $m$ bins of size $k$ with four items per bin if and
  only if there exists a K-tiling by $\tau$ with score at least $s$.

  \begin{figure}
    \centerline{\includegraphics[scale=.75]{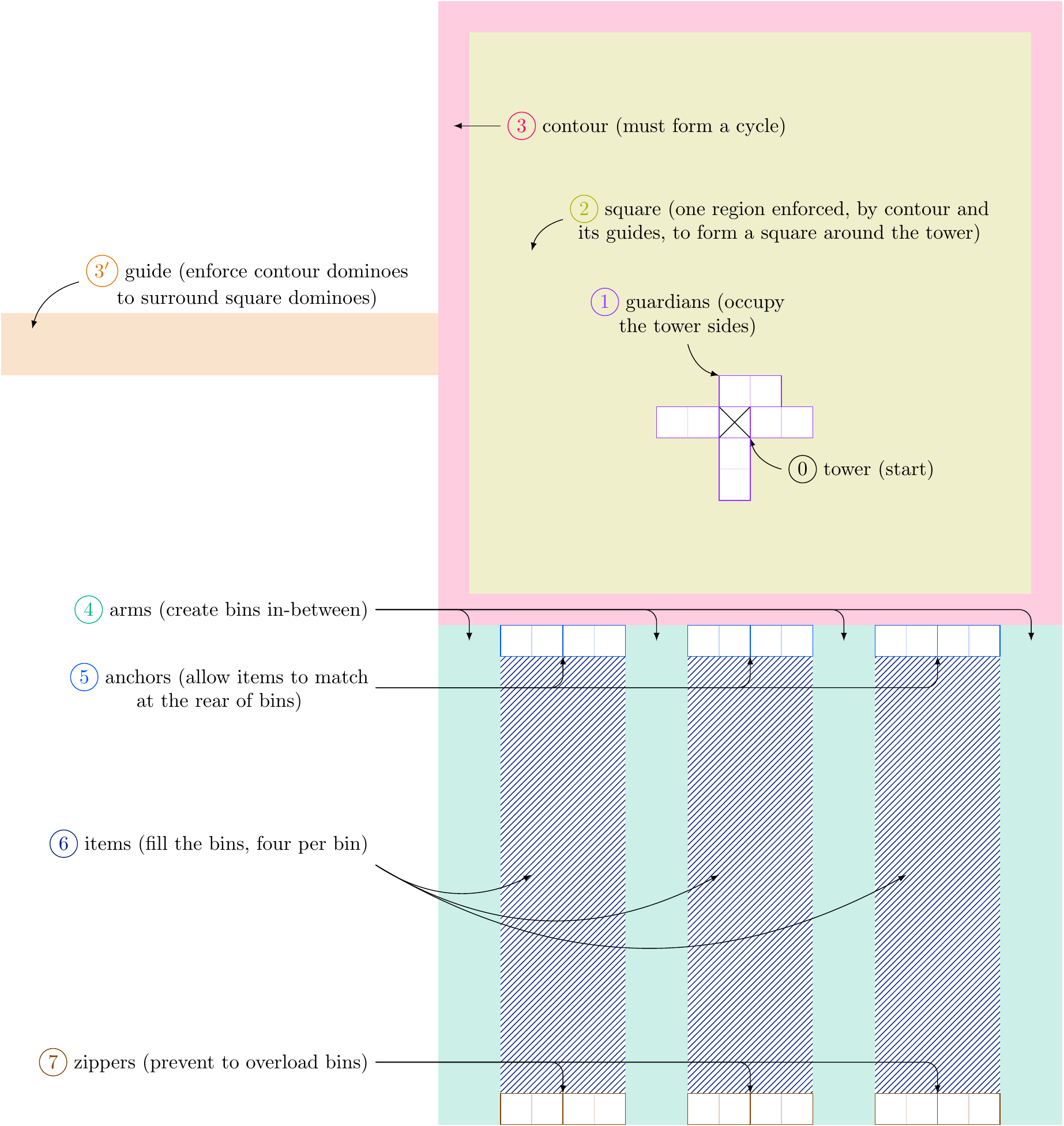}}
    \caption{
      Sketch of a K-tiling of score $s$ by $\tau$ in the reduction from {\bf
      4-Partition problem} to {\bf K-tiling problem}. Circled numbers from 1 to
      7 allow to follow the process of dominoes placement, chronologically. The
      main idea is to create $m=\frac{n}{4}$ bins of a given area (hatched on
      the figure), that can be filled with groups of {\em items} dominoes
      (whose quantities/areas correspond to item sizes) if and only if all items
      can be packed into the $m$ bins (with exactly four items, of sum $k$, per
      bin).
    }
    \label{fig:idea}
  \end{figure}

  \begin{figure}
    \centerline{\includegraphics[scale=.58]{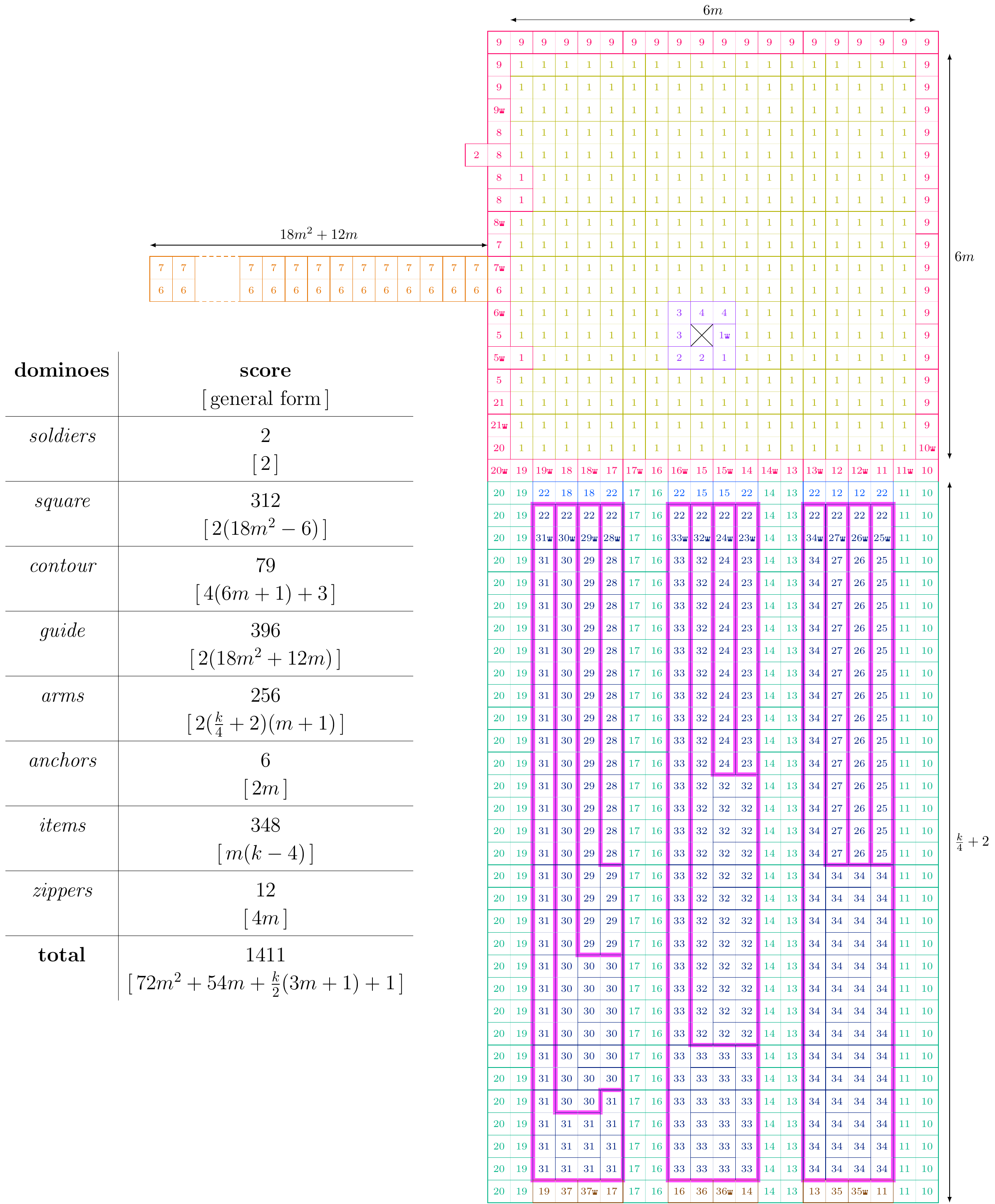}}
    \caption{To reduce the height of this figure, original sizes have only been
    multiplied by $4$ instead of $28$. A K-tiling by $\tau$ with score $s$
    (hence solving the {\bf K-tiling problem} instance), from a solution to the
    {\bf 4-Partition problem} instance with $n=12$, $m=3$, $k=120$ (originally
    $k=30$), and item sizes $12,12,16,16,16,16,24,40,40,48,48,72$ (originally
    $3,3,4,4,4,4,6,10,10,12,12,18$): first bag $72+16+16+16$, second bag
    $48+48+12+12$, third bag $40+40+24+16$. Domino colors are
    \textcolor{csoldiers}{\em guardians},
    \textcolor{csquare}{\em square},
    \textcolor{ccontour}{\em contour},
    \textcolor{cguide}{\em guide},
    \textcolor{carms}{\em arms},
    \textcolor{canchors}{\em anchors},
    \textcolor{citems}{\em items} (groups are \textcolor{chigh}{highlighted}), and
    \textcolor{czippers}{\em zippers}.
    The anchor color $3m+13$ (on which groups of item dominoes
    can match) equals $22$ on this example.}
    \label{fig:tiling}
  \end{figure}

  \medskip
  \fbox{$\Rightarrow$} Suppose there exists a packing of the $n$ items into the
  $m$ bins of size $k$ with exactly four items per bin, and let $X_j$ be the
  set of items in bin $j$. We construct the following K-tiling by $\tau$ (see
  Figure~\ref{fig:tiling}).
  \begin{enumerate}[topsep=.5em,itemsep=.5em]
    \item Place the four {\em guardians} dominoes around the tower.
    \item Around this create a square of size $6m \times 6m$ with the {\em
      square} dominoes $\hdomino{1}{1}$, leaving three dents empty on the left
      border of the square at the fifth, twelfth and thirteenth positions from
      the bottom left corner (the square has area $36m^2$, minus $9$ cells
      already taken by the {\em guardians} dominoes and the tower, minus $3$
      dents, hence exactly the $18m^2-6$ {\em square} dominoes are required).
    \item Make a path clockwise around the square with the {\em contour}
      dominoes in this order, filling the three dents with cells of color $1$,
      leaving the cell of color $2$ outside, and starting with the first dent
      at the fifth position above the bottom left corner of the square (the
      contour has length $4(6m+1)$, corresponding exactly to the $12m+4$ {\em
      contour} dominoes with four dents).
    \item Stack all {\em guide} dominoes on the left of the
      $\hdomino{7{\crown}}{6}$ domino of the border.
    \item Stack all {\em arms} dominoes, color by color, below the
      corresponding dominoes of the border. Observe that they match exactly one
      domino over three of the bottom border, creating $m+1$ stacks of length
      $\frac{k}{4}+2$, and therefore $m$ bins of size $k+8$ in between.
    \item Place a pair of {\em anchors} dominoes per bin, matching the existing
      colors, as on Figure~\ref{fig:tiling}.
    \item Place {\em items} dominoes corresponding to items of $X_j$ in bin
      $j$, filling $k$ cells of each bin and leaving the last row of four cells
      empty.
    \item Close each bin with the corresponding {\em zippers} dominoes (anchors
      {\em dominoes} fill four cells, consequently the {\em items} dominoes
      leave four cells in each bin, exactly the number of cells required for
      the {\em zipper} dominoes to match colors onto the arms on both
      ends\footnote{Note that the pattern of placement sketched on
      Figure~\ref{fig:tiling} can be extended to pack each bin with {\em items}
      dominoes corresponding to any four items of sizes summing to $k$, and
      leaving four cells on the bottom end for the {\em zippers} dominoes.}).
  \end{enumerate}
  The score of this K-tiling is $s$, as detailed on Figure~\ref{fig:tiling}.

  \medskip
  \fbox{$\Leftarrow$} This is the challenging part of the proof, where we will
  argue that the construction of a K-tiling by $\tau$ with score at least $s$
  is compelled to have the structure described above and illustrated on
  Figure~\ref{fig:tiling}, which corresponds to solving the {\bf 4-Partition
  problem} instance. The proofs of some claims are postponed.
  
  Suppose there exists a K-tiling by $\tau$ with score $s$. First notice that
  $s$ is an upper bound on the score one can obtain with a K-tiling by $\tau$,
  as it is the sum for each color of the number of cells of this color times
  the number of crowns on cells of this color. It must therefore correspond to
  a K-tiling with one region per color, except for colors $2,3,4$ and $3m+13$
  which have no crown. This will be the main assumption that will guide us as
  we study the dominoes chronologically. Also remark that all dominoes must be
  placed: at the beginning colors $2,3,4$ can always match the {\em tower}, and
  afterwards colors $2,3m+13$ appear only on dominoes with another color
  bringing some necessary points to the sum $s$.

  \begin{enumerate}[topsep=.5em,itemsep=.5em]
    \item There is no choice but to place the four {\em guardians} dominoes on
      the four sides of the tower. As a consequence, we don't have to treat the
      particular case of the tower anymore.
    \item All {\em square} dominoes are monochromatic, hence they form one
      large region of color $1$ (remark that they can, and therefore must, all
      be placed). We have now $9+2(18m^2-6)=36m^2-3$ cells occupied by
      some dominoes or the tower.
    \item For the {\em contour} dominoes the three cells of color $1$ must be
      connected to the unique region of color $1$ since this color will not
      appear anymore.
      {\em Contour} dominoes must be arranged in a cycle with some possible
      defects, that we define now. The idea is that colors $5$ to $8$ and $10$
      to $3m+12$ will have a simple path arrangement, whereas color $9$ may form
      a potato like connected component in which a path has to be identified.
      Let a {\em pseudo-cycle} be a K-tiling by {\em contour} dominoes such
      that, when cells of colors $1$ and $2$ are discarded:
      \begin{itemize}
        \item colors $5$ to $8$ and $10$ to $3m+12$ form a simple path
          (starting from the $5{\crown}$ cell of the first {\em contour}
          domino, the successor of each cell being either its partner domino
          cell, or a cell of the next domino, with the successor of the color
          $5$ from the last $\hhdomino{3m+12}{5}$ being the first $5{\crown}$
          cell), and
        \item a connected component of color $9$ connects two distinguished
          cells of color $9$: the $9{\crown}$ cell of domino
          $\hdomino{8}{9{\crown}}$, to the $9$ cell of domino
          $\hdomino{9}{10{\crown}}$.
      \end{itemize}
      The length of a pseudo-cycle is the length of its simple path plus the
      length of a shortest path (inside the connected component of color $9$)
      between the two distinguished cells of color $9$.
      %
      \begin{claim}
        \label{claim:contour}
        All {\em contour} dominoes must
        be placed (to reach score $s$), and they must form a pseudo-cycle,
        of length at most $4(6m+1)$.
      \end{claim}
      The pseudo-cycle is connected to the region of color $1$ via three dominoes
      $\hdomino{1}{5{\crown}}$, $\hdomino{8}{1}$, $\hdomino{8}{1}$:
      \begin{itemize}
        \item two of them
          are intended to frame the $\hdomino{6}{7{\crown}}$ domino,
        \item and the last
          one is for parity of cells number, since we have an odd number of
          occupied cells so far that is intended to form a square
          of even side length.
      \end{itemize}
      The pseudo-cycle of {\em contour} dominoes may have the region of
      color $1$ either inside its outer face,
      or inside its inner face (in this case
      the pseudo-cycle surrounds the region of color $1$).

    \item The {\em guide} dominoes enforce that the cycle of {\em contour}
      dominoes surrounds the region of color $1$. 
      \begin{claim}
        \label{claim:guidestack}
        The {\em guide} dominoes must be stacked one after the other in a straight
        segment, rooted at the analogous {\em contour} domino
      \end{claim}
      \begin{claim}
        \label{claim:guidesurrounds}
        The pseudo-cycle of {\em contour} dominoes must have the region
        of color $1$ inside its inner face.
      \end{claim}
      After this step we have $36m^2-3$ occupied cells surrounded
      by a cycle of $4(6m+1)$ cells with three additional cells (with
      color $1$) inside the cycle,
      hence $(6m)^2$ cells inside the cycle which is just long
      enough to make a square of side $6m+2$ around it. However, any other
      shape would either require a too long path, or leave a too small area
      inside, as stated in the following Claim.
      \begin{claim}
        \label{claim:square}
        The maximum area inside a cycle of $4(6m+1)$ cells contains $(6m)^2$ cells,
        and is achieve by a square shape of sides $6m+1$.
      \end{claim}
      \underline{\bf Intermediate conclusion:} at this point we have a square
      of {\em contour} cells with the tower, {\em guardians} and {\em square}
      dominoes inside, three dents of color $1$ inside, one dent of color $2$
      outside, and a stack of {\em guide} dominoes starting from the
      corresponding $\hdomino{6}{7{\crown}}$ {\em contour} domino (the reader
      can refer to Figure~\ref{fig:tiling} for an illustration).\\[.5em]
      We will argue thereafter why the contour is well aligned, with
      {\em contour} dominoes 
      $\hdomino{10}{11{\crown}}$ to
      $\hhdomino{3m+10}{3m+11{\crown}}$
      all on the same side of the square.
    \item The {\em arm} dominoes must create $m$ bins.
      \begin{claim}
        \label{claim:arms}
        The {\em arm} dominoes must be placed into $m+1$ bicolor stacks of length
        $\frac{k}{4}+2$ starting from the corresponding {\em contour} dominoes
        (separated by four positions), and joined with a pair of {\em zipper}
        dominoes.
      \end{claim}
      This also explains why {\em
      contour} dominoes are well aligned, with the third line of {\em contour}
      dominoes all on the same side of the square.\\[.5em]
      \underline{\bf Intermediate conclusion:} after these dominoes the
      K-tiling by $\tau$ with score $s$ must have created $m$ bins (with $m+1$
      arms) of size $4 \times (\frac{k}{4}+2)$. This size explains why the bin
      size (and consequently all item sizes) of the original {\bf 4-Partition
      problem} instance is converted to a multiple of $4$.
    \item Each pair of {\em anchors} dominoes must be placed so that each color
      already present in the contour (from $12$ to $3m+9$) form one region
      because these are the last dominoes with these colors. So there is a pair
      of {\em anchor dominoes} at the rear of each bin.
      Remark that color $3m+13$ has no crown hence it can be split into
      multiple regions. The purpose of this color is to be an anchor inside
      each bin, intended for groups of {\em items} dominoes to match.
    \item For each item $x_i$ we have a group of {\em items} dominoes, where the
      first domino of color $3m+13$ allows to match a bin anchor, and then all
      other dominoes of the group will form one region from this anchored
      domino (with the unique color $3m+13+i$ for each $i$), for a total
      of $x_i$ cells.
      \begin{claim}
        \label{claim:anchors}
        For each bin, at most four groups of {\em items} dominoes
        can match its anchors.
      \end{claim}
      \underline{\bf Intermediate conclusion:} As all $n$ groups of {\em items}
      dominoes must be placed on the board to reach score $s$, we must have
      exactly four groups of {\em items} dominoes in each of the
      $m=\frac{n}{4}$ bins (corresponding to the values of four items from the
      {\bf 4-Partition problem} instance).

    \item The {\em zippers} dominoes have the purpose of closing bins, with one
      pair of zippers matching the colors of each bin's arms. They must
      join the two arms of each bin with a path of length four (because of the
      unique pair of cell colors $3m+13+n+1$ to $4m+13+n$). However this is
      possible if and only if no {\em items} dominoes exceed a volume of $4
      \times (\frac{k}{4}+1)$ inside the bin (leaving the last row of four
      positions of each bin for the zipper), {\em i.e.} each bin contains four
      groups of {\em items} dominoes for a sum of at most $k$ cells (anchors
      already occupy four cells). Observe that when they contain a total of at
      most $\frac{k}{2}$ dominoes, it is always possible to place four groups
      of {\em items} dominoes in a bin and leave the last row for a pair of
      {\em zippers} dominoes (as on the example of Figure~\ref{fig:tiling}).
  \end{enumerate}

  \noindent\underline{\bf Conclusion:} to reach score $s$ with a K-tiling by
  $\tau$, a player must close the zipper on top of each of the $m$ bins and
  therefore trap inside each of them four items of sum at most $k$, for a total
  of $4m=n$ items, therefore solving the {\bf 4-Partition problem} instance.
\end{proof}


We now present the proofs of the different Claims, and recall their statements.

\setcounter{claim}{0}

\begin{claim}
  All {\em contour} dominoes must
  be placed (to reach score $s$), and they must form a pseudo-cycle,
  of length at most $4(6m+1)$.
\end{claim}

\begin{proof}[Proof of Claim~\ref{claim:contour}]
  For the simple path part, by induction on {\em contour} dominoes with some
  cell of color $5$ to $8$ or $10$ to $3m+12$, we have two simple paths
  concatenated by the last $\hhdomino{3m+12}{5}$ domino:
  \begin{itemize}
    \item either there is no occurrence of a color after the {\em
      contour} dominoes hence it must directly form one region (case of
      dominoes with color\footnote{The use of colors $1$ and $2$,
      though already present, changes nothing to this argument.} $8$,
      and color $5$ in the last {\em contour}
      domino which concatenates the two simple paths;
      {\em e.g.} when placing domino $\hhdomino{3m+12}{5}$
      its cell of color $5$ must match the existing region),
    \item or there is no other placed occurrence of one of its
      colors apart from the previous {\em contour} domino (case of all
      other {\em contour} dominoes;
      {\em e.g.} when placing domino $\hdomino{6}{7{\crown}}$
      its cell of color $6$ must match domino $\hdomino{5}{6{\crown}}$),
    \item or, for the $\hdomino{10}{11{\crown}}$ domino, we create the second
      simple path.
  \end{itemize}
  For the connected component of color $9$ the argument is straightforward,
  since the group of $\hdomino{9}{9}$ dominoes must form one region, and be
  connected to the two remaining cells of color $9$, in order to have a unique
  connected component of color $9$.
  For the length calculation, see Figure~\ref{fig:contour}.
\end{proof}

\begin{figure}
  \centerline{\includegraphics{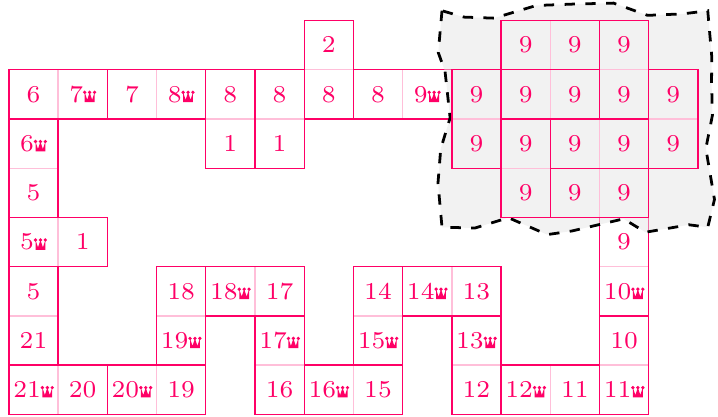}}
  \caption[{\em Contour} dominoes must form a pseudo-cycle]
  {{\em Contour} dominoes must form a pseudo-cycle,
  with a connected component of color $9$
  (dashed, containing the $9m-8$ $\hdomino{9}{9}$ dominoes)
  including a (shortest) path of length at most $2(9m-8)$,
  connecting the two extremities of a simple path containing
  all other cells (excluding cells of color $1$ and $2$).
  Domino by domino, the length of the simple path is therefore
  $1+2+2+2+1+1+1+2$ plus $2+2(3m+1)+4$, giving a total length of
  the pseudo-cycle upper bounded by $4(6m+1)$.}
  \label{fig:contour}
\end{figure}

\begin{claim}
  The {\em guide} dominoes must be stacked one after the other in a straight
  segment, rooted at the analogous {\em contour} domino
\end{claim}

\begin{proof}[Proof of Claim~\ref{claim:guidestack}]
  {\em Guide} dominoes are the last occurrences of colors $6$ and $7$, hence
  each of these colors must end up forming one region.
  The basic idea behind this proof is that when a $\hdomino{6}{7}$ domino is
  not correctly stacked it creates two regions of a color, which has to be
  reconnected via $\hdomino{6}{7}$ dominoes, which necessarily creates two
  regions of the other color, which has to be reconnected via $\hdomino{6}{7}$
  dominoes, which necessarily creates two regions of the other color, {\em
  etc}, hence the two colors cannot simultaneously end up forming one region.

  \medskip

  To formalize this intuition, we introduce some definitions.
  Let the {\em free outside region} denote the infinite set of grid cells with
  no domino on them ({\em free cells}), which are connected together.
  Given a K-tiling by some prefix of $\tau$ until some {\em guide} dominoes
  (between $1$ and $18m^2+12m$), let the {\em home-region} of color $c \in
  \{6,7\}$ be the region of the cell of color $c$ from the
  $\hdomino{6}{7{\crown}}$ {\em contour} domino.

  We say that a cell of color $c \in \{6,7\}$ is {\em zigzag-separated} from
  its home-region when:
  \begin{itemize}
    \item it does not belong to its home-region, and
    \item there does not exist a straight segment\footnote{A {\em straight
      segment} being a set of cells of the form
      $\{(x,y),(x,y)+e,(x,y)+2e,\dots,(x,y)+ze\}$ for some $x,y \in \Z$, $z \in
      \N$ and $e \in \{(0,1),(1,0),(0,-1),(-1,0)\}$.} of free cells and cells
      of color $c$, connecting it to its home-region.
  \end{itemize}
  We say that a cell of color $c \in \{6,7\}$ is {\em
  alternating-zigzag-separated} from its home-region when, for $\bar{c}=13-c$:
  \begin{itemize}
    \item it does not belong to its home-region, and
    \item any continuation of the K-tiling (placement of more $\hdomino{6}{7}$ {\em
      guide} dominoes) connecting it to its home-region would leave a cell of
      color $\bar{c}$ zigzag-separated from the home-region of color $\bar{c}$.
  \end{itemize}
  
  Let us now argue that, given a K-tiling by some prefix of $\tau$ until some
  {\em guide} dominoes (between $1$ and $18m^2+12m$):
  \begin{enumerate}[label={\it\alph*.}]
    \item\label{i:induction} if a cell of color $c \in \{6,7\}$ is
      zigzag-separated, then it is alternating-zigzag-separated,
    \item\label{i:base} if a {\em guide} domino is not correcly stacked as
      stated in the Claim, then a cell of color $c \in \{6,7\}$ is
      zigzag-separated.
  \end{enumerate}
  The combination of Items~\ref{i:base} (base case) and~\ref{i:induction}
  (induction) allows to conclude that it would be impossible to end up having
  one region of each color, since there would always exist a cell disconnected
  from its home-region. As a consequence, in order to reach score $s$, the
  Claim must hold.

  Item~\ref{i:base} follows a simple case disjunction presented on
  Figure~\ref{fig:guidestackbase}, and Item~\ref{i:induction} holds since in order
  to have one region of each color and reach score $s$, it is necessary to
  connect the zigzag-separated cell $p$ to its home-region. However, since this
  path is not a straight segment, it contains some turn which leaves a cell
  $\bar{p}$ of color $\bar{c}$ zigzag-separated from the home-region of color
  $\bar{c}$ (the path connecting $p$ to the home-region of color $c$ separates
  $\bar{p}$ from the home-region of color $\bar{c}$, see
  Figure~\ref{fig:guidestackinduction}).
\end{proof}

\begin{figure}
  \centerline{\includegraphics{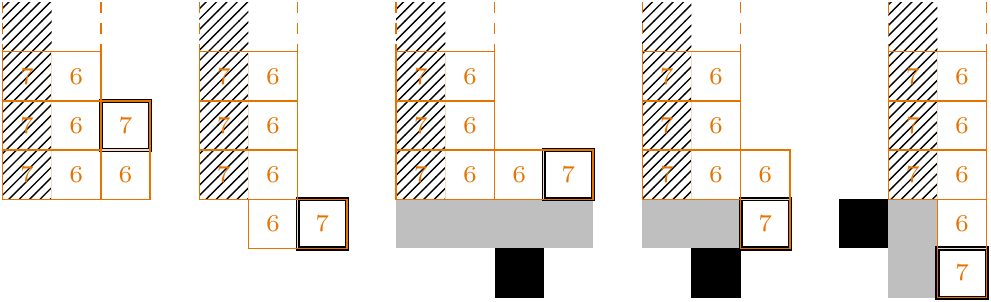}}
  \vspace*{.5cm}
  \centerline{\includegraphics{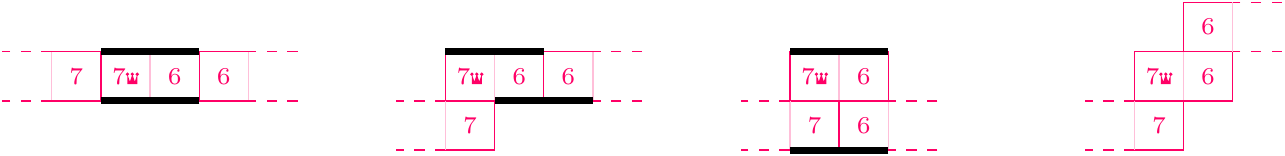}}
  \caption[guidestackbase]{
    Top: up to rotation, symmetry and swap of colors, five possible ways to
    not correctly stack a {\em guide} domino; in the two first cases a cell
    of color $7$ (bold) is zigzag-separated from its home-region (hatched);
    and in the three other cases, if $\hdomino{6}{7}$ {\em guide} dominoes
    are placed so that a straight segment (grey) connects the cell of color
    $7$ (bold) to its home-region (hatched), then there would necessarily
    exist a cell of color $6$ (black) zigzag-separated from its home-region.
    Hence in any case there exists a zigzag-separated cell.
    Bottom:  up to rotation and symmetry, possible stack positions are
    highlighted,
    depending on the placement of {\em contour} dominoes.
    For any stack position and number of correctly stacked
    dominoes, the argument on the top holds.
  }
  \label{fig:guidestackbase}
\end{figure}

\begin{figure}
  \centerline{\includegraphics[width=\textwidth]{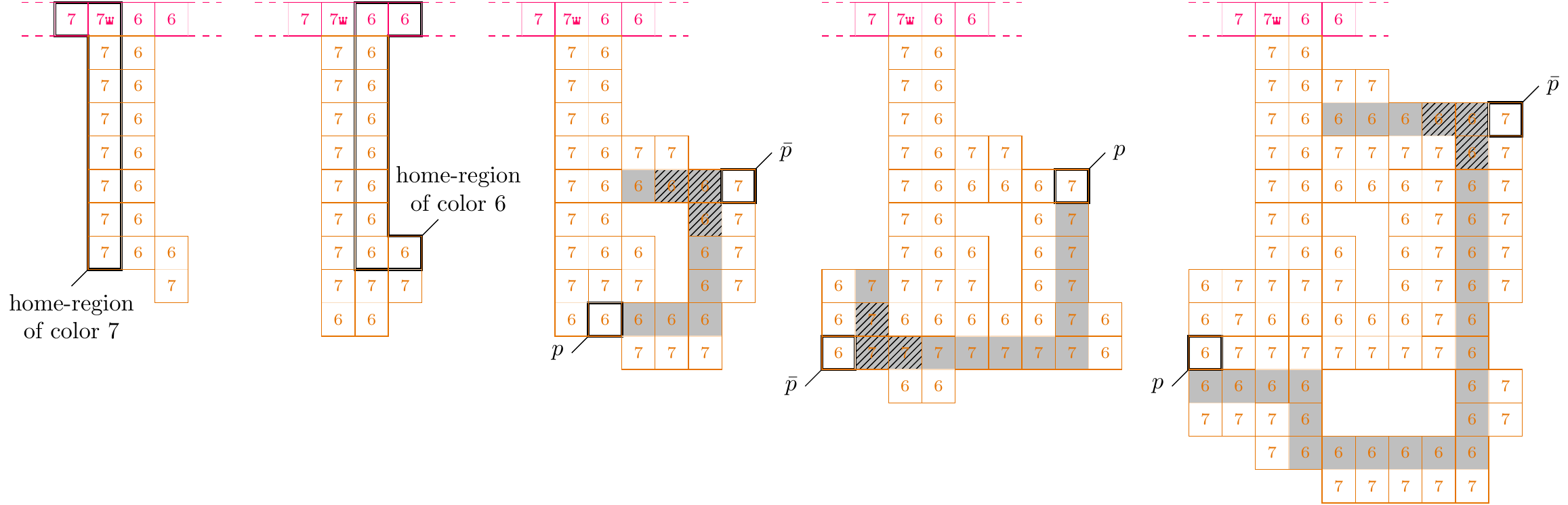}}
  \caption[guidestackinduction]{
    In order to reconnect a zigzag-separated cell $p$ of color $c \in
    \{6,7\}$ to its home-region,
    a zigzag-separated cell $\bar{p}$ of color $\bar{c}=13-c$ is
    created. Indeed, any path (grey) of cells of color $c \in \{6,7\}$
    connecting cell $p$ to its home-region acts as a
    barrier, and it contains some turn (dashed) which, due to the use of
    $\hdomino{6}{7}$ {\em guide} dominoes, necessarily leaves at least one
    cell $\bar{p}$ of color $\bar{c}$ zigzag-separated from the home-region
    of color $\bar{c}$. As a consequence, zigzag-separated cell $p$ is
    alternating-zigzag-separated.
  }
  \label{fig:guidestackinduction}
\end{figure}

\begin{claim}
  The pseudo-cycle of {\em contour} dominoes must have the region
  of color $1$ inside its inner face.
\end{claim}

\begin{proof}[Proof of Claim~\ref{claim:guidesurrounds}]
  Observe that if the
  pseudo-cycle of {\em contour} dominoes
  (Claim~\ref{claim:contour})
  has the region of color $1$ on its outer
  face, then the $18m^2+12m$ stacked $\hdomino{6}{7}$ {\em guide} dominoes
  (Claim~\ref{claim:guidestack})
  would have been placed inside
  the pseudo-cycle (at most $\frac{4(6m+1)}{2}$ of them) or inside the region of
  color $1$ (at most $18m^2-6$ of them), but they are too numerous
  so it would be impossible to have simultaneously a unique
  region for colors $6$ and $7$ (see Figure~\ref{fig:guide-number}).
\end{proof}

\begin{figure}
  \centerline{\includegraphics{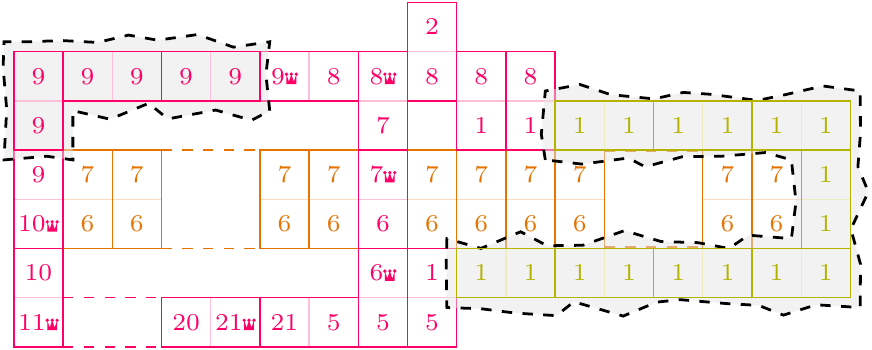}}
  \caption[The pseudo-cycle of {\em contour} dominoes must have the region of
  color $1$ on its inner face.]{If the pseudo-cycle of {\em contour} dominoes
  has the region of color $1$ on its outer face, then one cannot stack
  all the $18m^2+12m$ $\hdomino{6}{7}$ dominoes (the region of $1$
  contains $2(18m^2-6)$ cells, and the contour contains $4(m+1)$
  cells).}
  \label{fig:guide-number}
\end{figure}

\begin{claim}
  The maximum area inside a cycle of $4(6m+1)$ cells contains $(6m)^2$ cells,
  and is achieve by a square shape of sides $6m+1$.
\end{claim}

\begin{proof}[Proof of Claim~\ref{claim:square}]
  First remark that it is enough to consider rectangular shapes, because any
  {\em L-shape} croping some part of the inside area can be reversed to include
  this area instead of excluding it (hence increasing striclty the area inside
  the cycle). Then we have a rectangle of sides $a$ and $b$ (from $2$ to
  $2(6m+1)$), such that $2(a+b)=4(6m+1)$ is fixed, and we want to maximize its
  area, {\em i.e.} the product $ab$. It follows that $b=2(6m+1)-a$ and we want
  to maximize $a(2(6m+1)-a)$, a quadradic equation whose derivative reaches zero
  at $a=6m+1$, hence the area is maximum for a square shape of sides $6m+1$,
  containing $(6m)^2$ cells.
\end{proof}

\begin{claim}
  The {\em arm} dominoes must be placed into $m+1$ bicolor stacks of length
  $\frac{k}{4}+2$ starting from the corresponding {\em contour} dominoes
  (separated by four positions), and joined with a pair of {\em zipper}
  dominoes.
\end{claim}

\begin{proof}[Proof of Claim~\ref{claim:arms}]
  For each group of $\frac{k}{4}+2$ {\em arm} dominoes the argument is
  analogous to the proof of Claim~\ref{claim:guidestack},
  with two differences.
  \begin{itemize}
    \item For the base case, possible stack positions
      (Figure~\ref{fig:guidestackbase} for {\em guide} dominoes) are more
      restricted, due to the established square shape of {\em contour} dominoes
      (Claims~\ref{claim:contour} and~\ref{claim:square}) surrounding {\em
      square} dominoes (Claim~\ref{claim:guidesurrounds}).
    \item There is one extra cell of colors $11,13,14,16,\dots,3m+8,3m+10$ after {\em arm}
      dominoes, in {\em zippers} dominoes.
  \end{itemize}
  We now present how to adapt the argumentation structure from the proof of
  Claim~\ref{claim:guidestack}, in order to take into account these two differences.

  First, we can already notice that each pair of {\em zippers} dominoes contain
  a unique color (between $3m+13+n+1$ and $4m+13+n$), hence in order to have
  one region of each color and reach score $s$ they must form paths of length
  four, whose ends are connected to the regions of the respective colors
  (between $11$ and $3m+10$).

  Contrary to the four base cases presented on Figure~\ref{fig:guidestackbase}
  (bottom), only the two first cases are now possible, as presented on
  Figure~\ref{fig:armstackbase}, which also considers all ways to not correctly
  stack a first {\em arm} domino. It shows that there is only one case where an
  extra {\em zippers} cell may be useful, and that it leads to the
  impossibility to reach score $s$. As a consequence, the first domino of each
  arm must be correctly stacked.

  The induction is identical to the proof of Claim~\ref{claim:guidestack}: one
  can observe on Figures~\ref{fig:guidestackbase} (top)
  and~\ref{fig:guidestackinduction} that an extra {\em zippers} cell of color
  $6$ or $7$ is not enough to get one region of each color. As a consequence,
  for any couple of colors $3p+10,3p+11$ with $0 \leq p \leq m$, the sequence
  of $\frac{k}{4}+2$ $\hhdomino{3p+10}{3p+11}$ {\em arm} dominoes must also be
  arranged as a stack starting from the corresponding {\em contour} dominoes.

  The remark first made about {\em zippers} corresponds to the
  second part of the Claim.
\end{proof}

\begin{figure}
  \centerline{\includegraphics[width=\textwidth]{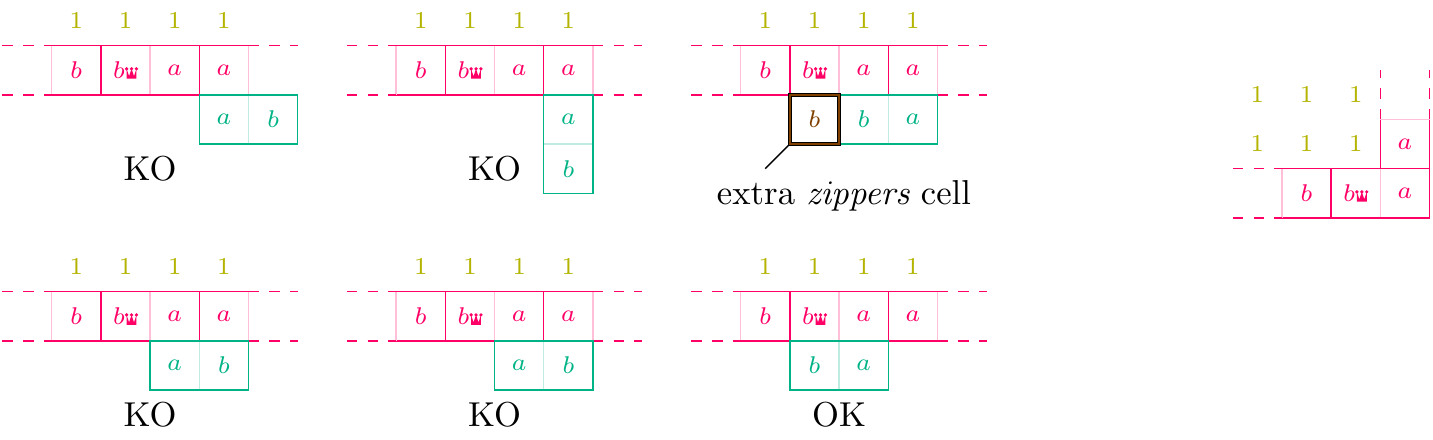}}
  \vspace*{1.5cm}
  \centerline{\includegraphics[width=\textwidth]{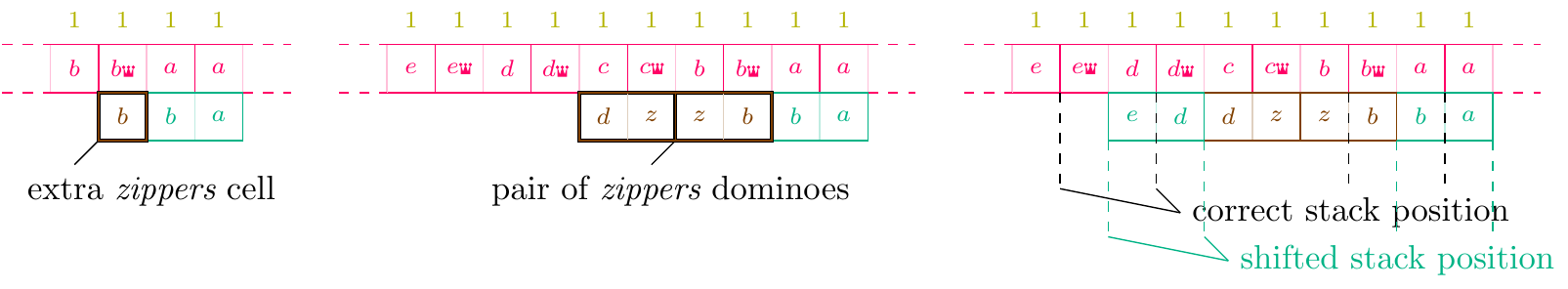}}
  \caption[{\em Arm} dominoes must be correctly stacked.]{
    Presentation of the base case for Claim~\ref{claim:arms}.
    Top and middle: given that {\em contour} dominoes form a square surrounding {\em
    square} dominoes (Claims~\ref{claim:contour}, \ref{claim:guidesurrounds}
    and~\ref{claim:square}), there are only two cases to consider for the
    positionning of {\em arm} dominoes.
    The first six configurations (left) present, up to rotation, symmetry and swap of
    colors, all placements of the first dominoe in the first case for some
    arm domino of colors $a,b$ with $a \in \{10,13,16,\dots,3m+10\}$ and
    $b=a+1$. One can see that an extra {\em zippers} cell of color $b$ may be
    useful to get one region of each color $a$ and $b$ only in one
    configuration (shown in the bottom part).
    For the others, four marked ``KO'' require more {\em arm} dominoes, and
    one marked ``OK'' corresponds to a correctly stacked {\em arm} domino.
    The second case is presented on the last configuration (right),
    any placement of a first
    {\em arm} domino other than the intended one
    leads to an impossibility by a reasoning analogous to
    the first cases (note that in this second case there is no color symmetry).
    Bottom: when a first {\em arm} domino is not correctly stacked but we can
    get one region of each color $a$ and $b$ using the extra {\em zippers}
    cell of color $b$ (left), the couple of {\em zippers} dominoes (via
    color $z=3m+13+n+\frac{a-7}{3}$ appearing only in this couple)
    must be connected to color
    $d=c+1$, with $c=b+1$, on the other end (middle). This is possible only
    with an {\em arm} domino of the next group ($e=d+1$) also not correctly
    stacked (right, with the same ``shift'' of the stack of {\em arm}
    dominoes, sketched). However, this prevents the two {\em anchors}
    dominoes containing cells of color $c$ to match {\em contour} dominoes of color
    $c$, hence these two {\em anchors} dominoes are lost and score $s$ cannot be reached.
  }
  \label{fig:armstackbase}
\end{figure}

%

\begin{claim}
  For each bin, at most four groups of {\em items} dominoes
  can match its anchors.
\end{claim}

\begin{proof}[Proof of Claim~\ref{claim:anchors}]
  Each item has size at
  least $28$ and therefore corresponds to at least $14$ dominoes, however
  after placing a pair of {\em anchors} dominoes and four of these
  minimum size
  groups of $14$ {\em items} dominoes in any possible way, no anchor cell
  of color $3m+13$ is available for a fifth group of {\em items} dominoes
  (see details on Figure~\ref{fig:anchors}).
  This argument explains why all item and
  bin sizes of the original {\bf 4-Partition problem} instance have been
  multiplied by $28$ (and not simply by $4$): so that each group of {\em
  items} dominoes is large enough to enforce that at most four per bin can
  match the anchor.
\end{proof}

\begin{figure}
  \centerline{\includegraphics[width=\textwidth]{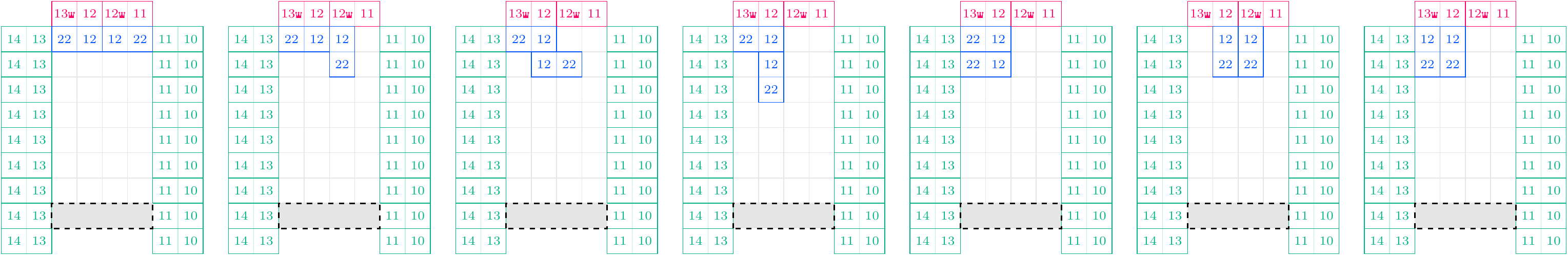}}
  \caption{Up to axial symmetry, seven different ways to place a pair of
  anchors in a bin (colors of the first bin are taken as an example).
  After placing the first group of at least $14$ {\em items} dominoes, at
  least one position on the eighth row of the bin (dashed) is occupied
  (one can simply count available positions); after placing the second
  group, at least a second position of the eighth row is occupied; after
  the third group a third one; and after the fourth group the eighth row
  of the bin is full of {\em items} dominoes. However, after these four
  groups of {\em items} dominoes, cells of anchor color $3m+13$ cannot
  exceed the seventh row (the third row after the pair of {\em anchors}
  dominoes, plus one for each group of {\em items} dominoes),
  consequently no more group of {\em items} dominoes can match an anchor
  color and take place inside the bin.}
  \label{fig:anchors}
\end{figure}

\section{Conclusion}
\label{s:conc}

Theorem~\ref{theorem:main} establishes that {\em Kingdomino\tm} shares the
feature of many fun games: it requires to solve instances of an $\NP$-complete
problem. Finding efficient moves is therefore\footnote{Unless $\Poly=\NP$.}
a computationally hard
task, and players may feel glad to encounter good solutions.

As we have seen in Section~\ref{s:count}, the number of possible K-tilings may
grow rapidly. The main difficulty in designing of the $\NP$-hardness
reduction to the {\bf K-tiling problem}, has been to find an initial sequence
of dominoes which imposes a rigid structure (with very few possible K-tiling
reaching a maximum score), and still allows to be continued in order to
implement some strong $\NP$-complete problem (given by the instance from the
reduction).

\bigskip

Our modeling of the game {\em Kingdomino\tm} abstracts various aspects of the
game (as board games are finite, this is necessary), and our construction in
Theorem~\ref{theorem:main} is frugal in terms of crowns, but it is opulent in
terms of colors (we have not tried to diminish the usage of colors).
An open question is whether the {\bf K-tiling problem} is
still $\NP$-hard if the number of colors is
bounded?

\bigskip

It would also be interesting to integrate the multi player notion of strategy
to the abstract modeling of the game.
This somewhat involved process in the
official {\em Kingdomino\tm} rules may be simplified
as follows.
Given,
\begin{itemize}
  \item $k$ K-tilings $K_1,\dots,K_k$ by some sequences of dominoes
    $\tau_1,\dots,\tau_k$, and
  \item an unordered set of dominoes $\tau$,
\end{itemize}
the $k$ players
\begin{enumerate}
  \item construct $k$ sequences of dominoes $\tau'_1,\dots,\tau'_k$
    by picking one domino at a time from $\tau$
    (turn by turn, until $\tau'_1,\dots,\tau'_k$ form a partition of $\tau$),
    and then
  \item each player $P_i$ plays its sequence $\tau'_i$ from $K_i$.
\end{enumerate}
We ask whether player $i \in \{1,\dots,k\}$ has a winning
strategy.

Remark that this modeling also discard two rules of the official {\em
Kingdomino\tm} game. First, at the domino picking stage, the order in which
players pick dominoes is fixed, and does not depend on the previously picked
dominoes as in the official game. Second, the game is separated in two stages:
a domino picking stage where players construct their $\tau'_i$ until $\tau$ is
empty, and then a domino placement stage where players place their $\tau'_i$,
whereas in the official game these alternate. Also note that the strategy part
of this multi player game is on the first stage only, the second stage only
consists for each player $P_i$ to maximize its score given $K_i$ and $\tau'_i$.

For any $k$ this problem is solvable in polynomial-space ($\PSPACE$),
because one can enumerate all possible game plays (all sequences of dominoes
choices $\tau'_1,\dots,\tau'_k$, and for each of them the best achievable score
of each player) and discover whether player $P_i$
has a winning strategy or not (the existence of a winning strategy
corresponds to the satisfiability of a quantified propositional formula).

Note that, similarly to other multi players games, starting from
empty boards (only the tower for each player, {\em i.e.}
$\tau_1=\dots=\tau_k=\emptyset$), a strategy stealing argument would lead to
the conclusion that the first player always has a winning
strategy\footnote{For two players: by contradiction suppose $P_2$ has a winning
strategy; $P_1$ first takes any domino $d$ and then follows the winning
strategy of $P_2$ on $\tau$ ({\em i.e.} $P_1$ picks dominoes according to the
choices of $P_2$ as if she had never taken domino $d$ and $P_2$ had started the
game); if at some point $P_1$ needs to take domino $d$ (according to the
strategy being stolen), then she takes any available domino $d'$ and the
reasoning goes on with $d$ substituted by $d'$. Player $P_1$ can always steal the moves of $P_2$, and therefore (by hypothesis) construct a sequence $\tau'_1$ leading to a score higher than $\tau'_2$.}.
As a consequence, adding non-empty starting boards $K_1,\dots,K_k$ is necessary,
and hopefully makes a hardness proof easier to construct.
However players' boards are independent of each other,
in the sense that each player plays its sequence of dominoes only
on {\em its own board},
which makes the setting a bit different from $\PSPACE$-hardness results
encountered in the literature about multi players games, such as
{\em Hex}~\cite{hex1,hex2},
{\em Checkers}~\cite{checkers},
{\em Go}~\cite{go} 
and other two players games with perfect information~\cite{twoplayerperfectinfo}.

\section*{Acknowledgments}

The work of K\'evin Perrot has mainly been funded by his salary of French State agent, assigned to
Univ. C{\^o}te d'Azur, CNRS, Inria, I3S, UMR 7271, Sophia Antipolis, France,
and to
Aix Marseille Univ., Univ. Toulon, CNRS, LIS, UMR 7020, Marseille, France.
We received auxiliary financial support from
the \emph{Young Researcher} project ANR-18-CE40-0002-01 FANs,
project ECOS-Sud C16E01, and
project STIC AmSud CoDANet 19-STIC-03 (Campus France 43478PD).

\bibliographystyle{plain}
\bibliography{biblio}

\end{document}